\documentclass[runningheads]{llncs}
\usepackage{url}
\usepackage[hidelinks]{hyperref}
\usepackage[utf8]{inputenc}
\usepackage{graphicx}
\usepackage{amsmath}
\usepackage{graphicx}
\usepackage{amssymb}
\usepackage[tableaux]{prooftrees}
\usepackage{pifont}
\usepackage{subcaption}
\usepackage{dsfont}
\usepackage{todonotes}
\usepackage{multirow}
\usepackage{bm}
\usepackage{tabularx}
\usepackage{bbding}

\renewcommand{\paragraph}[1]{\par\smallskip\noindent\textbf{#1}}

\usepackage[shortlabels]{enumitem}
\newcommand{\subsumes}[0]{\mathrel{\ooalign{\hss$\subseteq$\hss\cr\kern0.7ex\raise0.45ex\hbox{\scalebox{0.65}{$\sigma$}}}}}

\def\orcidID#1{\href{http://orcid.org/#1}{\raisebox{-1.25pt}{\includegraphics{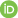}}}}

\usepackage{pifont}

\newcommand{\vampire}{\mbox{\textsc{Vampire}}}
\newcommand{\eprover}{\textsc{E}}
\newcommand{\spass}{\mbox{\textsc{Spass}}}
\newcommand{\setheo}{\mbox{\textsc{Setheo}}}
\newcommand{\leancop}{\mbox{\textsf{leanCoP}}}
\newcommand{\cadical}{\mbox{\textsc{CaDiCaL}}}
\newcommand{\zzz}{Z3}
\newcommand{\meanCoP}{\textsf{meanCoP}}
\newcommand{\upCoP}{\mbox{\textsc{UPCoP}}}

\usepackage{bbding}

\begin{document}
\title{Finding Connections via Satisfiability Solving}
\author{Clemens Eisenhofer\inst{1}\textsuperscript{(\Envelope)}\orcidID{0000-0003-0339-1580} \and
Michael Rawson\inst{2}\orcidID{0000-0001-7834-1567} \and
Laura Kov\'{a}cs\inst{1}\orcidID{0000-0002-8299-2714}}
\authorrunning{Eisenhofer et al.}
%
\institute{TU Wien, Vienna, Austria\\
\email{\{clemens.eisenhofer,laura.kovacs\}@tuwien.ac.at}
\and
University of Southampton, Southampton, UK\\
\email{michael@rawsons.uk}
}
\maketitle
\begin{abstract}
    Commonly used proof strategies by automated reasoners organise proof search either by ordering-based saturation or by reducing goals to subgoals. In this paper, we combine these two approaches and advocate a SAT-based method with symmetry breaking for connection calculi in first-order logic, with the purpose of further pushing the automation in first-order classical logic proofs.
    In contrast to classical ways of reducing first-order logic to propositional logic, our method encodes the structure of the proof search itself. We present three distinct SAT encodings for connection calculi, analyse their theoretical properties, and discuss the effect of using SAT/SMT solvers on these encodings. 
    We implemented our work in the new solver \upCoP{} and showcase its practical feasibility.
\end{abstract}

\section{Introduction}
Search strategies employed by \emph{automated theorem provers for first-order logics} can be divided into two broad classes~\cite{taxonomy}: \emph{ordering-based} and \emph{subgoal-reduction}. 
The first class, which contains saturation-based theorem provers  including \vampire~\cite{Vampire}, \eprover~\cite{eprover}, and \spass~\cite{spass}, work by continuously deducing new facts from an existing set of formulas and expanding the search space with these new facts. 
The second class, containing systems such as \setheo~\cite{SETHEO} or \leancop~\cite{leanCoP}, works by manipulating a partial proof and implementing backtracking if necessary.

The subgoal-reduction class has the disadvantage that redundant, ``useless", formulas in the search space may be explored in duplicate manner, unless special care is taken to ``remember'' where the prover has already been before.
Avoiding such redundant cases for the purpose of efficient reasoning 
is a subject of great interest. Therefore,  \emph{global refinement} within the subgoal-reduction approach to theorem proving has been proposed and investigated~\cite{comparison-of-proof-methods}.
In general, such refinements can contain non-trivial propositional structure; for example, the information ``if clauses $C$ and $D$ are in the current proof attempt, and the current substitution binds $x \mapsto t$ and $y \mapsto s$, we are in a dead-end and have to backtrack'' is a hard formula to reason with.

Backtracking mechanisms are routinely implemented in Boolean satisfiability (SAT) solvers~\cite{DBLP:conf/sat/BiereFW23}, for the purpose of refining the (partial) proof.  
Modern SAT solvers \emph{learn} relevant information as they go, and even allow users to add constraints during the solver's search for a model, in response to the solver's current (partial) assignment.
When a SAT solver cannot find a satisfying assignment, an \emph{explanation} in the form of an unsat core is given. In this paper, we examine these learning and explanation features in the context of first-order theorem proving, which usually makes SAT solvers an ideal vehicle for managing global information and their refinements.

\emph{Here we are interested in the integration of SAT solving and subgoal-reduction}, focusing on Bibel's \emph{connection method}~\cite{Bibel}.
Connection methods yield powerful calculi for automation of expressive logics, for example, for classical first-order logic~\cite{Bibel,SETHEO}, higher-order logic~\cite{DBLP:journals/jar/Andrews89}, linear logic~\cite{Galmiche2000}, and several modal logics~\cite{DBLP:conf/lpar/Otten12}.
Our work directly encodes the search for connection proofs as a Boolean satisfiability problem, allowing the solver to dictate search decisions and respond by asserting constraints, such that when a satisfying assignment is reached, it represents a complete proof. To this end, we introduce three encodings of connections (Sections~\ref{sec:tableau}--\ref{sec:unsatCore}) and show their complementary nature and power (Section~\ref{sec:implementation}).

First, we present our method applied to connection \emph{tableaux} (Section~\ref{sec:tableau}) and highlight some unfortunate properties of this setting. Therefore, we refine our encoding of the connection calculus using a matrix form (Section~\ref{sec:matrix}). Further, we describe how unsat cores can be used to guide iterative deepening during SAT solving (Section~\ref{sec:unsatCore}). The resulting encodings allow many global refinements that are usually not feasible within other methods (Section~\ref{sec:optimisations}).

Our work intends for the SAT solver to return a satisfying assignment of our constraints, where the model encodes a finished proof: matrix or tableau. \emph{In other words, we represent a first-order proof as a propositional model.}
This contrasts with most other uses of SAT solvers in theorem proving in which ground unsatisfiability is the aim~\cite{AVATAR}, often witnessing Herbrand-style refutation by instantiation of first-order clauses. 

The practical use of our approach is  showcased by evaluating our work on the TPTP repository of first-order problems~\cite{TPTP}. To this end, we provide a proof-of-concept implementations of our techniques in our new solver \upCoP{}, where \upCoP{} uses either the \cadical{} SAT solver~\cite{CaDiCal} and the \zzz{} SMT solver~\cite{DBLP:conf/tacas/MouraB08}. When evaluating \upCoP{} against the state-of-the-art \meanCoP{} solver~\cite{curiously-effective}, our experimental findings demonstrate that our work can solve $179$ problems that \meanCoP{} cannot (Section~\ref{sec:implementation}).

\paragraph{Summary of Contributions.}
\begin{enumerate}
    \item We encode the existence of first-order connection calculi proofs as Boolean satisfiability problem, using connection tableaux (Section~\ref{sec:tableau}), a matrix representation (Section~\ref{sec:matrix}), and iterative deepening via unsat core refinement (Section~\ref{sec:unsatCore}). We show soundness and completeness of our encodings.
    \item We explore optimisations to reduce redundancy and symmetry in the encodings (Section~\ref{sec:optimisations}). 
    \item We discuss our implementation of the encodings in the new prototype solver \upCoP{} and evaluate it on the TPTP benchmark set (Section~\ref{sec:implementation}). 
\end{enumerate}

\section{Preliminaries}
\label{sec:preliminaries}
We use standard syntax and semantics of classical first-order logic~\cite{smullyan}.
Logical objects such as terms $t$ may be indexed: $t^i_j$.
We assume the input problem has been negated and converted to conjunctive normal form (CNF) by a satisfiability-preserving transformation~\cite{handbook-ar-nf}. However, this transformation is optional in theory~\cite{Bibel}.

\subsection{Satisfiability Solving}
We assume familiarity with Boolean satisfiability (SAT) solving~\cite{handbook-of-satisfiability} and satisfiability modulo theories (SMT)~\cite{DBLP:conf/jelia/Tinelli02}.
In addition to the basic decision procedure for Boolean formulas, many SAT solvers support solving \emph{under assumptions} and \emph{unsatisfiable cores}.
Solving under assumptions allows fixing some literals temporarily for the duration of a solving run: afterwards, the solver ``forgets'' them and their consequences.
If the solver detects that the problem is unsatisfiable under assumptions, it may extract a subset of the assumptions used to derive inconsistency: the so-called ``unsat core''.
Cores may not be \emph{minimal}, so inconsistency can be derived with a strict subset of the core.
Minimal cores can be generated at additional computational cost~\cite{minimal-unsat-cores-smt,minimal-unsat-cores}.

Some SAT and SMT solvers, including \cadical~\cite{ipasir-up} and \zzz~\cite{user-propagation}, allow the user to intervene \emph{during} search by a variety of means, often under the title ``user propagation''.
Such mechanisms allow employing a solver to tackle a broad class of problems efficiently.
For our purposes, we assume we can be notified when a SAT variable is assigned true or false, and respond by asserting additional constraints, potentially containing fresh SAT variables.
We write $J_1, \ldots, J_n \Vdash F$ to represent that we added (\emph{propagated}) the constraint $J_1 \wedge \ldots \wedge J_n \Rightarrow F$ for some formula $F$ to the solver, given that the solver's current model satisfies all antecedent literals $J_1, \ldots, J_n$.
This feature allows us to avoid eagerly generating an extensive set of all possible constraints and add only the currently relevant parts of the encoding. This kind of lazy generation is desirable in our case.

\subsection{Connection Tableaux}
\label{sec:prelim_tableaux}
\begin{figure}[t]
    \centering
    \begin{forest}
        [,
        [$L$], [\ldots], [$K$]
        ]
    \end{forest}
    \hspace{.1in}
    \begin{forest}
        [\ldots,
        [$L'$, name=root
            [$L$, name=mate]
            [\ldots]
            [$K$]
        ]
        ]
        \draw (mate) to[out=north west, in=west] (root);
    \end{forest}
    \hspace{.1in}
    \begin{forest}
        [\ldots,
        [$L'$, name=ancestor
        [\ldots,
        [$L$, name=leaf], [\ldots]
        ], [\ldots]
        ]
        ]
        \draw (leaf) to[out=north west, in=west] (ancestor);
    \end{forest}
    \caption{Connection tableau rules, left-to-right: \emph{start}, \emph{extension}, and \emph{reduction}. In \emph{start} and \emph{extension}, $L \vee \ldots \vee K$ is a freshly-renamed copy of a clause from the input problem. In \emph{extension} and \emph{reduction}, $L$ is connected to $L'$ using $\sigma$.}
    \label{fig:tableau-rules}
\end{figure}
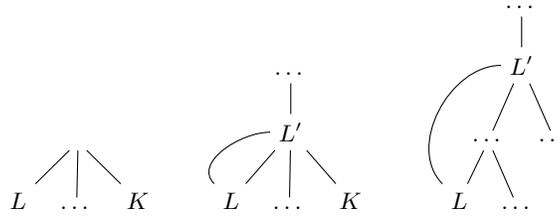
Connection tableaux are essentially clausal tableaux~\cite{tableauxHandbook} with the additional constraint that each clause added to a branch must have at least one literal \emph{connected} to the current leaf literal~\cite{handbook-ar-model-elimination}. \emph{Two literals are connected if they have the same atom but opposite polarity: they are dual.}
Recall that in the first-order case, clauses in the tableau have their variables renamed apart from any other, and a global substitution $\sigma$ is applied to the entire tableau to connect literals (unification).
The connection tableau calculus is not \emph{confluent} and requires both backtracking and a fair enumeration of tableaux for completeness.

We say that two literals $L, K$ \emph{can} be connected and write $L \bowtie K$ if there exists some substitution $\rho$ such that $\rho(L)$ is connected to $\rho(K)$.
A subset of input clauses are considered potential roots of the tableau~\cite{handbook-ar-model-elimination}: we assume these \emph{start} clauses (conjectures) have been chosen in a way that at least one conjecture can be used for finding a proof.
Equality is not handled by the basic connection calculus, and it is either axiomatised~\cite{handbook-ar-paramodulation} or preprocessed away by some variation of Brand's modification~\cite{brand}.
We sometimes write $C^k$ to distinguish the $k$\textsuperscript{th} copy of the input clause $C$, indexing its variables as $x^k$.

Conventionally,  three operations manipulate connection tableaux, shown in Figure~\ref{fig:tableau-rules}.
\emph{Start} operations pick a start clause and add it at the root of the tableau. The chosen clause's literals are the initial leaves of the tableau.
\emph{Extension} operations add a copy of a clause from the input clause set below a leaf literal of the tableau, connecting at least one of the newly added clause's literals with the leaf.
\emph{Reduction} operations connect a leaf literal with another literal on the path from the literal toward the tableau's root.
A branch of the tableau is \emph{closed} in case the leaf of that branch is connected to some other literal. The whole tableau is closed, thus representing a proof if all of its branches are closed. In general, all these operations must be backtracked over to achieve completeness.

\subsection{The Connection Method, Matrices, and Spanning Connections}
\label{sec:matrixDescr}
Connection tableaux are an instance of the connection method~\cite{Bibel}.
While connection calculi are a rich topic with many facets, we are primarily interested in the following \emph{matrix} representation. 
We consider matrices in \emph{normal form} and thus define a matrix as a non-empty set of rigid clauses; that means, clauses containing rigid variables.
The matrix can contain an arbitrary number of rigid copies of the same input clause, where variables are renamed apart, and a global substitution $\sigma$ is applied. A \emph{path} through a matrix is a set that contains exactly one literal from each copy of the clause in the matrix and is called \emph{closed} if it contains at least one connected pair of literals, otherwise \emph{open}.
A \emph{matrix proof}, or a matrix with a \emph{spanning set of connections}, is a matrix for which there are no open paths.

Further, a matrix is \emph{fully connected} with respect to a set of connections if each literal in the matrix is connected to at least one other literal of a different clause~\cite{using-matings-for-pruning}.
A matrix proof $M$ is \emph{minimal} if there is no proof using only a strict subset of $M$. 
Although we can require that there is at least one conjecture in the matrix we could consider as a start clause, there is, in contrast to connection tableau proofs, no inherent tree structure in matrix proofs.

\begin{example}
\label{ex:running}
To illustrate the two representations, consider the unsatisfiable set
\[ \left\{~~ \forall x\forall y.~ \lnot P(x) \vee \lnot P(f(y)),~~ \forall z.~ P(z) \vee P(f(z)) ~~\right\}
\]
and compare matrix and tableau refutations thereof in Figure~\ref{fig:tab} and~\ref{fig:mat}\footnote{the literals of each clause are written vertically in columns and the order of the clause copies does not matter}. Figure~\ref{fig:openPath} shows a fully connected matrix that is not a proof as it has an open path.
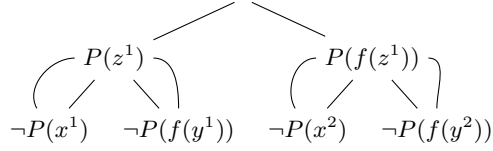
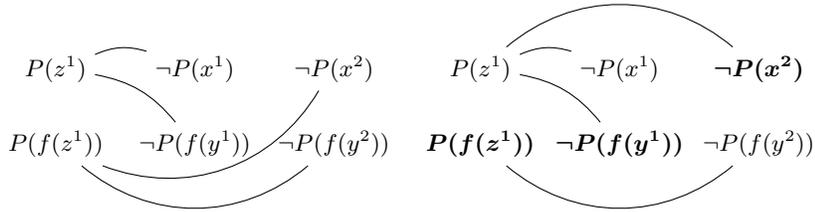
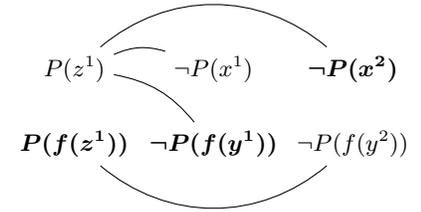
\begin{figure}[t]
\begin{center}
\begin{subfigure}[t]{.7\textwidth}\centering
\begin{forest}
[
 [$P(z^1)$, name=pz
  [$\lnot P(x^1)$, name=px1]
  [$\lnot P(f(y^1))$, name=pfy1]
 ]
 [$P(f(z^1))$, name=pfz
  [$\lnot P(x^2)$, name=px2]
  [$\lnot P(f(y^2))$, name=pfy2]
 ]
]
\draw (px1) to[out=120, in=west] (pz);
\draw (pfy1) to[out=north, in=east] (pz);
\draw (px2) to[out=120, in=west] (pfz);
\draw (pfy2) to[out=north, in=east] (pfz);
\end{forest}
\caption{Tableau proof with curved lines indicating connections. $\sigma$ is computed such that, e.g., $\sigma(z^1) = f(y^1)$.}
\label{fig:tab}
\end{subfigure}
\end{center}
\begin{subfigure}[t]{.45\textwidth}\centering
\begin{tikzpicture}[main/.style = {draw, circle}] 
\node (C31) at (0, 0) {$P(z^1)$}; 
\node (C32) at (0, -1) {$P(f(z^1))$}; 

\node (C11) at (1.85, 0) {$\lnot P(x^1)$}; 
\node (C12) at (1.85, -1) {$\lnot P(f(y^1))$}; 

\node (C21) at (3.7, 0) {$\lnot P(x^2)$}; 
\node (C22) at (3.7, -1) {$\lnot P(f(y^2))$}; 

\path (C11) edge[bend right=20] (C31);
\path (C12) edge[bend right=20] (C31);

\path (C21) edge[bend left=40] (C32);
\path (C22) edge[bend left=40] (C32);

\end{tikzpicture}
\caption{Same proof in matrix form}
\label{fig:mat}
\end{subfigure}
\begin{subfigure}[t]{.45\textwidth}\centering
\begin{tikzpicture}[main/.style = {draw, circle}] 
\node (C31) at (0, 0) {$P(z^1)$}; 
\node (C32) at (0, -1) {$\bm{P(f(z^1))}$}; 

\node (C11) at (1.85, 0) {$\lnot P(x^1)$}; 
\node (C12) at (1.85, -1) {$\bm{\lnot P(f(y^1))}$}; 

\node (C21) at (3.7, 0) {$\bm{\lnot P(x^2)}$}; 
\node (C22) at (3.7, -1) {$\lnot P(f(y^2))$}; 

\path (C11) edge[bend right=20] (C31);
\path (C12) edge[bend right=20] (C31);

\path (C21) edge[bend right=40] (C31);
\path (C22) edge[bend left=40] (C32);
\end{tikzpicture}
\caption{A matrix that does not represent a proof. An open path is shown in \textbf{bold}.}
\label{fig:openPath}
\end{subfigure}
\caption{Connection calculus derivations}
\end{figure}
\end{example}

\section{Encoding Connection Tableaux}
\label{sec:tableau}
We first encode the search for closed connection tableaux (Section~\ref{sec:prelim_tableaux}) in a SAT solver. 
A closed connection tableau is explicitly constructed from a satisfying assignment.
We encode that a literal $L$ is a leaf of the tableau at path $U$ using a SAT variable $\langle L; U\rangle$, where $U$ is a set of literals labelling the nodes from $L$ towards the root.
$L$ and $U$ are used only to determine the corresponding variable. For example, $\lnot P(x^2)$ in Figure~\ref{fig:tab} is represented by a variable $\langle \lnot P(x^2); \{ P(f(z^1)) \}\rangle$. 
The substitution $\sigma$ and unification of connected literals 
are handled with another family of variables we discuss later. 

\paragraph{\bf Connection tableau rules as SAT.} We begin by asserting that at least one start clause $S$ must be present in the tableau. Therefore, all literals $L \in S$ must be in the tableau at the root:
\begin{equation}
\label{eq:start}
\bigvee_S \bigwedge_{L \in S} \langle L; \emptyset \rangle
\end{equation}
The SAT solver is free to choose any start clause $S$, but all literals in the chosen  $S$ must be present at the root of the tableau.
As the solver assigns variables $\langle L; U \rangle$ true
we respond by propagating additional requirements.
We  demand that each literal has either an \emph{extension} $E_{C, K}$ or a \emph{reduction} $R_{K}$ applied, in order to close the corresponding branch in the final tableau:
\begin{equation}
\label{eq:tableauxEncoding}
\langle L; U \rangle \Vdash \bigvee_{C, K} E_{C, K} \vee \bigvee_{K \in U} R_{K}
\end{equation}
Each formula $E_{C, K}$ represents applying an extension operation at $L$ using a fresh copy of a clause $C$ containing a literal $K \bowtie L$, which yields
\begin{equation}
\label{eq:ext}
E_{C, K} := \left(\langle L \sim K\rangle \wedge \bigwedge_{\substack{K' \in C\\K' \neq K}} \langle K'; \{ L \} \cup U \rangle\right)
\end{equation}
i.e. that if an extension step $E_{C, K}$ is taken, $L$ and $K$ are connected and other literals $K' \in C$ must be in the tableau with path $\{ L \} \cup U$.
We write $\langle L \sim K \rangle$ for the SAT variable representing that $L$ and $K$ are connected modulo $\sigma$. Similarly,
\begin{equation}
\label{eq:red}
R_K := \langle L \sim K\rangle
\end{equation}
where $K \bowtie L$ and $K \in U$.
The steps $E_{C, K}$ and $R_{K}$ are computed based only on the \emph{possible} connection relation $\bowtie$: the current global substitution $\sigma$ maintained by the SAT solver is ignored, otherwise relevant cases may be missed on backtracking.
Iterative deepening may be applied as usual~\cite{leanCoP}, for example, by offering no $E_{C, K}$ options if the path length $|U|$ exceeds the depth limit.

\begin{example}
    Consider the input clauses from Example~\ref{ex:running} and assume the solver assigned literal $\langle \lnot P(x^2); \{ P(f(z^1)) \}\rangle$ true, so the solver considers literal $\lnot P(x^2)$ to be within the tableau, below $P(f(z^1))$. The solver propagates
\begin{align*}
    & \langle \lnot P(x^2); \{ P(f(z^1)) \}\rangle \Vdash\\
    &\hspace*{5em}\left(\langle P(f(z^2)); \{ P(f(z^1)), \lnot P(x^2) \} \rangle \wedge \langle\lnot P(x^2) \sim P(z^2)\rangle\right) \vee \vphantom{A}\\
    &\hspace*{5em}\left(\langle P(z^2); \{ P(f(z^1)), \lnot P(x^2) \} \rangle \wedge \langle\lnot P(x^2) \sim P(f(z^2))\rangle\right) \vee \vphantom{A}\\
    &\hspace*{5em}\langle\lnot P(x^2) \sim P(f(z^1))\rangle
\end{align*}
where the first two disjuncts represent applications of the extension rule and the last an application of reduction.
\end{example}

\paragraph{{\bf Unification Constraints.}} 
Variables $\langle L \sim K\rangle$ constrain $\sigma$ such that $\sigma(L)$ is connected to $\sigma(K)$.
When the SAT solver assigns such a literal, we check whether this is consistent with the existing set of constraints.
This can be done by applying a unification algorithm, perhaps using an efficient data structure such as the \emph{variable trail}~\cite{handbook-ar-model-elimination} to handle backtracking.
We note in passing that algebraic datatype solvers~\cite{datatypes}, which are supported by many SMT solvers, implement a similar decision procedure.
If the constraints are not satisfiable, we produce a \emph{conflict clause} containing the reasons as Boolean assignments.
For example, if we have $\langle L \sim K\rangle$, $\langle J \sim K\rangle$ and $\langle L \sim J\rangle$, but $\langle L \sim K\rangle \wedge \langle L \sim J\rangle$ is already unsatisfiable, we add the conflict
\begin{equation}
\lnot\langle L \sim K\rangle \vee \lnot\langle L \sim J\rangle
\end{equation}
causing the solver to backtrack.
This approach also allows a uniform treatment of refinements such as \emph{regularity} based on \emph{disequation constraints}~\cite{handbook-ar-model-elimination}.

\paragraph{\bf SAT Encoding of Closed Connection Tableaux.} We now have all the ingredients for our SAT encoding, which we denote by $\mathcal{E}_T$. By asserting that (i) a start clause must be present~\eqref{eq:start}, (ii) each literal in the tableau must have a reduction or extension rule applied to it~\eqref{eq:tableauxEncoding}
and (iii) connections must have a consistent unifier, enforced by unification constraints, our encoding $\mathcal{E}_T$ is finished. Each propositional model of $\mathcal{E}_T$ represents a closed connection tableau.

\paragraph{\bf Pathological Behaviour.}
Our SAT encoding $\mathcal{E}_T$ has drawbacks. 
Importantly, extension adds a \emph{fresh instance} from the clause set to the tableau, and so the number of different SAT variables $\langle L; U \rangle$ grows rapidly.
In turn, this means the resulting SAT problem has only limited propositional structure between variables that the solver can exploit.
Search tends to degrade towards the kind of exhaustive enumeration of possible derivations that systems such as \leancop{}~\cite{leanCoP} efficiently implement, but with the added overhead of a SAT solver.

\section{Encoding Matrix Proofs}
\label{sec:matrix}
To avoid the previously described problems of $\mathcal{E}_T$, we encode matrix proofs (see Section~\ref{sec:matrixDescr}). We denote our matrix-based encoding $\mathcal{E}_M$. 
Most search routines for spanning sets of connections presented in literature~\cite{Bibel,leanCoP} restrict connections such that proofs are, or could be, simulated by a connection tableau proof~\cite{matrix-based-constructive}.
In our $\mathcal{E}_M$ encoding, we allow arbitrary connections between clauses in the matrix.
In this way, a single matrix proof can correspond to numerous tableau proofs~\cite{using-matings-for-pruning}.
In any event, our new representation $\mathcal{E}_M$ produces a combinatorial problem of finding connections between a set of clauses, which we argue is much more suitable for SAT solvers than $\mathcal{E}_T$.

\subsection{Encoding Overview}
We find a matrix proof with a given \emph{resource limit} and split reasoning in: 
\begin{enumerate}
    \item We encode constraints for a fully-connected matrix (Section~\ref{sec:fully-connected}).
    \item We constrain that the result has a set of spanning connections (Section~\ref{sec:check}).
\end{enumerate}
We use the following result to motivate our encoding.
\begin{theorem}[Fully Connected Matrix]
\label{thm:connected}
Suppose $M$ is a minimal matrix proof with a spanning set of connections. Then $M$ is fully connected.
\end{theorem}
\begin{proof}
First, note that this is similar but not identical to Proposition~1 in Letz's work on matings pruning~\cite{using-matings-for-pruning}.
Suppose, towards contradiction, there is a literal $L \in C \in M$ that is not connected to any other $K$.
Now consider the rest of the matrix $M' = M \setminus \{ C \}$.
Since $M$ is minimal, there is an open path $U$ through $M'$, as otherwise $M'$ would have a spanning set of connections.
As $L$ is not connected to any literal,  $U \cup \{ L \}$ is an open path for $M$, which conflicts with our assumption of $M$ being a proof.
\end{proof}
Theorem~\ref{thm:connected} allows us to restrict our work to fully-connected matrices. This restriction is a good approximation, as few fully-connected matrices do not have a spanning set of connections in practice (see Section~\ref{sec:experiments}).
We use SAT variables of the form $S_C$ to denote that clause $C$ appears in the matrix, sometimes superscripted $S_C^k$ to indicate selecting $C^k$, the $k^{th}$ copy of $C$. We call these $S_C$ \emph{selectors} and call $C$ \emph{selected} if $S_C$ is assigned true.
At least one of the start clauses $C$ must be selected, cf.~\eqref{eq:start}: 
\begin{equation}\label{eq:start:matrix}
\bigvee S^1_C. 
\end{equation}
In Section~\ref{sec:tableau}, we apply iterative deepening on the maximum length of a branch. This kind of resource limit cannot be used here, as there is no obvious notion of a \emph{branch}, so we must devise alternatives. We first apply iterative deepening on the matrix's number $d$ of clauses. We immediately see that we need to introduce at most $d$ selectors for each clause. As we always refer to the same copies, the solver can more easily learn sensible conflicts from $\mathcal{E}_M$ than from $\mathcal{E}_T$. We discuss a further enhanced encoding later in Section~\ref{sec:unsatCore}.

\subsection{Fully Connected Matrices}
\label{sec:fully-connected}
By Theorem~\ref{thm:connected}, we may constrain that each literal in the matrix must connect to at least one other literal.
Similarly to \eqref{eq:tableauxEncoding}, we respond to a selection $S_C$ by propagating that each literal must be connected to some other literal in another clause in the matrix by enumerating all possible connections.
This other clause could be selected or require selection, but there is no distinction between extension and reduction.
Suppose $C$ is selected. For each $L \in C$, we propagate
\begin{equation}
\label{eq:matrixConnect}
    S_C \Vdash~\bigvee_D~\bigvee_{1 \le k \le d}\bigvee_{K \in D^k} S^k_D \wedge \langle L \sim K\rangle
\end{equation}
where $K \bowtie L$ is a literal in the input clause $D$ we connect to, and $k$ indicates which copy $D^k$ of that clause is used.

There are several possible options to enforce that at most $d$ clauses are selected for the matrix.
We suggest using pseudo-Boolean constraints~\cite{DBLP:journals/tcad/ChaiK05} or a direct encoding~\cite{DBLP:conf/cp/BailleuxB03,handbook-of-satisfiability,DBLP:conf/cp/Sinz05} to constrain that ``there are no more than $d$ selector variables assigned''.
We can strengthen this to \emph{exactly} $d$ as we apply iterative deepening, so the less-than-$d$ case was encountered already.

\begin{example}
    \label{ex:connection}
    Assume the situation 
    \[ \{~ S^1_{P(z) \vee P(f(z))},~ \lnot S^1_{\lnot P(x) \vee \lnot P(f(y))},~ \langle P(z^1) \sim \lnot P(x^1) \rangle,~ \langle P(z^1) \sim \lnot P(f(y^1)) \rangle ~\}\] 
    which corresponds to the situation in Figure~\ref{fig:openPath} without the 3\textsuperscript{rd} clause. Assuming that there are at most two copies of each clause, we would propagate that $P(f(z^1))$ needs to connect somehow. That means,
    \begin{align*}
        &S^1_{P(z) \vee P(f(z))} &\Vdash &&\left(S^1_{\lnot P(x) \vee\lnot  P(f(y))} \wedge \langle P(f(z^1)) \sim \lnot P(x^1) \rangle\right) \vee \vphantom{A}\\
        &&&&\left(S^1_{\lnot P(x) \vee \lnot P(f(y))} \wedge \langle P(f(z^1)) \sim \lnot P(f(y^1)) \rangle\right) \vee \vphantom{A}\\
        &&&&\left(S^2_{\lnot P(x) \vee \lnot P(f(y))} \wedge \langle P(f(z^1)) \sim \lnot P(x^2) \rangle\right) \vee \vphantom{A}\\
        &&&&\left(S^2_{\lnot P(x) \vee \lnot P(f(y))} \wedge \langle P(f(z^1)) \sim \lnot P(f(y^2)) \rangle\right) \hphantom{\vee} 
    \end{align*}
    These are all possible options given the available clauses and their copies.
\end{example}

\subsection{Spanning Sets of Connections}
\label{sec:check}
Once we have a fully connected matrix, we check for open paths (see Section~\ref{sec:matrixDescr}).
If there are none, we are done and can use the resulting SAT model to output the proof, consisting of the selected clauses, the connections, and the global substitution.
Suppose instead there is an open path $U$ -- given by the set of literals -- through the matrix $M$.
At least two literals along $U$ must connect to make $M$ a proof.
Let $\bar{S}$ be the set of selectors assigned true.
Propagating
\begin{equation}
\label{eq:openpath1}
\bar{S} \Vdash \bigvee_{\{L, K\} \subseteq U} \langle L \sim K\rangle
\end{equation}
forces the solver to ``fix'' $M$, likely via backtracking, by requiring that $U$ is not an open path.

\begin{example}
    \label{ex:spanning}
    Continuing Example~\ref{ex:connection}, assume the SAT solver further assigns 
    \[ \{~ S^2_{\lnot P(x) \vee \lnot P(f(y))},~ \langle P(f(z^1)) \sim \lnot P(f(y^2)) \rangle, \langle P(z^1) \sim \lnot P(x^2) \rangle ~\}\] 
    true. The result is the fully connected matrix, which is shown in Figure~\ref{fig:openPath}. We cannot propagate further clauses that are not already satisfied. However, we can find the open path $\{~ P(f(z^1)),~ \lnot P(f(y^1)), \lnot P(x^2) ~\}$,
    which we exclude with
    \begin{align*}
        &S^1_{P(z) \vee P(f(z))}, S^1_{\lnot P(x) \vee \lnot P(f(y))}, S^2_{\lnot P(x) \vee \lnot P(f(y))} \Vdash\\
        &\hspace*{5em}\langle P(f(z^1)) \sim \lnot P(f(y^1)) \rangle \vee \langle P(f(z^1)) \sim \lnot P(x^2) \rangle.
    \end{align*}

    As neither $\langle P(f(z^1)) \sim \lnot P(f(y^1)) \rangle$ nor $\langle P(f(z^1)) \sim \lnot P(x^2) \rangle$ results in a consistent unifier, the SAT solver is required to backtrack.
\end{example}

\subsection{Correctness and Complexity of Matrix Encodings}
Encoding $\mathcal{E}_M$ consists of~\eqref{eq:start:matrix},~\eqref{eq:matrixConnect},~\eqref{eq:openpath1}, and constraints for the depth limit.
Its models represent fully-connected matrix proofs.
We show soundness, completeness, and termination for a given size $d$ in $\mathcal{E}_M$, and describe the respective complexity class of $\mathcal{E}_M$.

\begin{theorem}[Soundness]
\label{thm:sound}
A propositional model of $\mathcal{E}_M$ represents a matrix with a spanning set of connections.
\end{theorem}
\begin{proof}
Whenever the SAT solver finds a propositional model, we first check that it represents a proof, adding constraints if not (Section~\ref{sec:check}).
\end{proof}

\begin{theorem}[Completeness]
\label{thm:complete}
If a matrix $M$ together with a spanning set of connections exists, there is a 
propositional model of $\mathcal{E}_M$ at depth $d = |M|$.
\end{theorem}
\begin{proof}
Such a matrix proof $M$ can be represented by setting $S^k_C$ true iff there are at least $k$ copies of $C$ in $M$.
The spanning set of connections is represented by setting $L \sim K$ iff $L$ is connected to $K$ in the proof.
This model of $\mathcal{E}_M$ and all its submodels are consistent modulo the semantics of $\sim$ and all possible instances of~\eqref{eq:matrixConnect}. Furthermore, the final model satisfies the depth constraints and contains at least one start clause.
Also, we do not block the model with the final check in Section~\ref{sec:check} as the set of connections are spanning.
\end{proof}

\begin{theorem}[Complexity Bound]
\label{thm:complexity}
Solving our particular encoding $\mathcal{E}_M$ is in the complexity class $\Sigma_2^P$ with respect to the input size and the matrix proof size.
\end{theorem}
\begin{proof}
There are polynomially many SAT variables.
To see this, let $c$ be the number of clauses in the input, containing a total of $l$ literals.
We have at most $d\cdot c$ selectors $S^k_C$.
We also have $O(d^2l^2)$ possible connection literals $\langle L \sim K\rangle$.
Hence, there are only polynomially many instantiations of~\eqref{eq:matrixConnect}. After adding in the worst case, all of them, the problem is in NP. We can non-deterministically guess an assignment for all polynomially many selectors and unification atoms. Checking the model can be done clearly in deterministic polynomial time.
Checking whether a matrix represents a matrix proof is co-NP complete. It can be solved by a separate SAT solver, which checks if the matrix $\sigma(M)$ represented by the SAT model is satisfiable. As we can solve $\mathcal{E}_M$ in NP with a co-NP oracle, the problem of solving our encoding for some fixed limit $d$ is in $\Sigma_2^P$.
\end{proof}
As checking the satisfiability of a set of clauses over rigid variables is $\Sigma_2^P$-complete~\cite{DBLP:conf/stacs/Goubault94}, our approach's complexity coincides with its theoretical bound.
\begin{corollary}[Termination]
A run for solving $\mathcal{E}_M$ at fixed $d$ terminates.
\end{corollary}

\section{
Iterative Deepening via Unsat Core Refinement}
\label{sec:unsatCore}
A downside of our encoding $\mathcal{E}_M$, especially of its constraints from Section~\ref{sec:fully-connected}, is that we eagerly introduce and use selectors for clause instances that are not required.
If there is more than one input clause and the matrix is of size $d$, not all clauses can have $d$ copies in the matrix for arithmetic reasons.
Therefore, creating $d$ instances of each clause is overkill. This section addresses this challenge and improves iterative deepening via unsat cores, resulting in a refined encoding 
$\mathcal{E}_U$.  

We use an abstraction-refinement~\cite{cegar} approach to approximate the number of copies required for each clause.
This way, we avoid polluting the search space with likely unnecessary clause instances.
Instead of a coarse global limit $d$, we estimate how many copies of each clause are required with a \emph{multiplicity} $\mu$~\cite{matings-in-matrices}. Initially, we have $\mu(C) = 1$ for start clauses and $\mu(C) = 0$ otherwise.
The multiplicity is monotonically increased based on the unsat core of the following encoding.
We refine constraint~\eqref{eq:matrixConnect} to 
\begin{equation}
    \label{eq:matrixConnect2}
    S_C \Vdash~\bigvee_D~\bigvee_{1 \le k \le \bm{\mu(D) + 1}}~\bigvee_{K \in D^k} S^k_D \wedge \langle L \sim K\rangle
\end{equation}
as we have $\mu(D)$ copies of $D$.
Note that $k$  ranges up to $\mu(D) + 1$.
We add temporary assertions\footnote{named $\kappa$ because it indicates that a clause needs more ``$\kappa$-city''} $\kappa_D := \lnot S_D^{\mu(D) + 1}$ so that the solver cannot select $D^{\mu(D) + 1}$, but \emph{can} report that finding a proof failed in part due to a lack of copies of $D$.
We revise~\eqref{eq:openpath1}, as we can no longer assume that a fully connected matrix has exactly $d$ clauses. 
A matrix can be fully connected, but the desired matrix proof may in fact be a strict \emph{superset} of the current matrix.
As~\eqref{eq:openpath1} is now too strong, we weaken it to
\begin{equation}
\label{eq:openpath2}
\bar{S} \Vdash \left[\bigvee_{\{L, K\} \subseteq U} \langle L \sim K\rangle\right] \bm{\vee \bigvee_{L \in U} F_L}
\end{equation}
where $F_L$ is a formula indicating that $L$ could also be connected to another literal in a clause \emph{not yet in the matrix} and $\bar{S}$ a set of selectors as before in~\eqref{eq:openpath1}.

\begin{example}
    Assume we have $\mu(P(z) \vee P(f(z))) = \mu(\lnot P(x) \vee \lnot P(f(y))) := 2$ and let us consider Example~\ref{ex:connection} again. In this propagation, we would have to additionally list the case that $S^3_{\lnot P(x) \vee \lnot P(f(y))}$ is true. The case $S^3_{P(z) \vee P(f(z))}$ is not needed, as there is no way to connect it anyway. As we assumed that $S^3_{\lnot P(x) \vee \lnot P(f(y))}$ is false, the SAT solver will not assign it true, but might report it in its unsat core in case there are feasible ways to proceed, given a higher depth limit. In Example~\ref{ex:spanning}, we would have to list the additional cases:
    \[ S^2_{P(z) \vee P(f(z))} \wedge \langle \lnot P(f(y^1)) \sim P(z^2)\rangle,~ S^2_{P(z) \vee P(f(z))} \wedge \langle \lnot P(f(y^1)) \sim P(f(z^2))\rangle,~\]
    \[S^2_{P(z) \vee P(f(z))} \wedge \langle \lnot P(x^2) \sim P(z^2)\rangle,~ S^2_{P(z) \vee P(f(z))} \wedge \langle \lnot P(x^2) \sim P(f(z^2))\rangle,~\]
    \[ S^3_{P(z) \vee P(f(z))},~ S^3_{\lnot P(x) \vee \lnot P(f(y))} \]
    As the respective clause selector is assumed false anyway, we do not need to introduce connectedness atoms for clause copies beyond the depth limit. Slight variations and optimisations of the encoding, however, would benefit from encoding those aspects as well.
    Further, listing both the second and third copies of clause $P(z) \vee P(f(z))$ shows one source of symmetry discussed in Section~\ref{sec:optimisations}.
\end{example}

Whenever the SAT solver reports unsatisfiability, we retrieve the \emph{unsat core} representing a potentially non-minimal subset of $\kappa$ assertions sufficient to yield unsatisfiability.
We may increase one or more $\mu(C)$ if the corresponding assertion occurs in the unsat core.
However, to retain completeness, we need to ensure that we eventually increment the multiplicity of every clause appearing repeatedly in the unsat core: in other words, we require \emph{fairness}.
In case the core is empty, we can conclude that no proof exists.
As a result, 
our SAT encoding $\mathcal{E}_U$ with improved iterative deepening is given by~\eqref{eq:start:matrix},~\eqref{eq:matrixConnect2}, and~\eqref{eq:openpath2}.
\begin{example}
\label{ex:infExt}
Consider the input problem
\begin{align*}
&C := P(a) &&D := \forall x.~\lnot P(x) \vee P(f(x)) &&E := \forall y.~\lnot P(y)
\end{align*}
with $C$ as the start clause. $\kappa_{D}$ will always be contained within the unsat core, no matter its multiplicity.
However, a fair enumeration eventually includes $\kappa_{E}$.
\end{example}
Our improved encoding $\mathcal{E}_U$ remains sound and terminating by similar arguments to Theorems~\ref{thm:sound}~and~\ref{thm:complexity}. Completeness requires an adjusted argument.
\begin{theorem}[Completeness]
\label{thm:completeness2}
If a matrix $M$ together with a spanning set of connections exists, there is a corresponding propositional model of $\mathcal{E}_U$.
\end{theorem}
\begin{proof}
In addition to Theorem~\ref{thm:complete}, we show that if there is a proof using $M$ which our current $\mu$ does not permit, at least one relevant $\kappa_C$ is contained in the unsat core.
Fairness then ensures we will eventually find the proof.
Consider a maximal subset $M' \subset M$ representable at $\mu$.
\begin{enumerate}
\item If $M'$ cannot be fully connected, it contains at least one literal $L$ with no connections. Since $M'$ is maximal and $M$ can be fully connected, $L$ should be connected to some literal in a clause $D$ that is not yet in the matrix.
This option is offered in~\eqref{eq:matrixConnect2}, but fails
because the respective $\kappa_D$ assumption is forced false. $\kappa_D$ is therefore in the unsat core.
\item If $M'$ can be fully connected, we would have failed to close some open path $U$ and propagated some instance of~\eqref{eq:openpath2}.
Some $L \in U$ must connect to at least one literal of a clause not yet in $M'$ by the right disjunct of~\eqref{eq:openpath2}.
As $M'$ is maximal, we can add no clauses, so the constraint fails because of the $\kappa$ assumption.
\end{enumerate}
\end{proof}
A side effect of encoding $\mathcal{E}_U$ is that it also terminates on some non-theorems.


\section{Redundancy Elimination}
\label{sec:optimisations}
Restricting the SAT solver's search space is beneficial when solving the SAT encodings of Sections~\ref{sec:tableau}--\ref{sec:unsatCore}. In addition to standard techniques,  such as tautology elimination~\cite{handbook-ar-model-elimination}, we propose some specialised redundancy eliminations.

\subsection{Multiplicity Symmetry}\label{sec:sb}
Our encodings from Sections~\ref{sec:tableau}--\ref{sec:unsatCore} contain \emph{several symmetries}~\cite{DBLP:journals/tc/AloulSM06}, which we now \emph{avoid}, rather than \emph{break}~\cite{symmetry-breaking-fmb}.
The first symmetry is that copies of clauses are interchangeable.
Suppose we connect some literal $L$ to literal $K$ in a copy of $C$ not yet in the matrix, and subsequently fail to find a proof in that direction.
Nothing prevents the SAT solver from selecting another copy of $C$ and failing for the same reasons.
We avoid this by propagating
\begin{equation}
    \label{eq:clauseIndexOrder}
    S^{i + 1}_C \Vdash S^i_C, 
\end{equation}
enforcing that $C^i$ is selected only if  $C^j$ with $j < i$ are selected.

\subsection{Subsumption and Instance Symmetry}
Saturation systems often delete a clause $C$ because it is \emph{subsumed}~\cite{term-indexing} by some more-general clause $D$.
Dynamics in connection systems are somewhat different as new first-order clauses are not deduced, but we can profit by applying some subsumption.
If two different clauses $C$ and $D$ are in the current matrix, we can enforce that neither becomes a subset of the other, modulo $\sigma$\footnote{we do not apply an additional substitution to either side}.
This restriction preserves completeness, by Bibel's Lemma 6.8~\cite{Bibel}.

An obvious extension of this idea is removing clauses from the matrix that are subsumed by other clauses from the input set. This fails, however.
\begin{example}
    Consider the four input clauses
    \[ C := P(a) ~~~~ D := Q(a) ~~~~ E := \forall x.~\lnot P(x) \vee Q(x) ~~~~  F := \forall y.~\lnot Q(y) \]
    with $C$ being the only start clause. There is a proof without subsumption via $C$, $E$, and finally $F$; this is the only minimal proof using $C$. However, putting $E$ in the matrix with $\sigma(x) = a$ makes it be subsumed by $D$.
\end{example}
Subsumption in the usual sense of smaller clauses representing any usage of larger clauses fails. This failure occurs because we might lose the \emph{reason to connect} a clause to our current matrix. Keeping larger clauses instead also does not work, as we might not be able to connect all literals of the larger clause.
Nonetheless, we can motivate additional symmetry avoidance this way.
Define an arbitrary total order $\prec$ on input clauses such that start clauses are the least elements. We assume that the order of each clause in the matrix is the same as the order of the clauses in the input set from which they are a copy.

\begin{lemma}[Instance Symmetry]
\label{lem:symmetry}
Suppose there is a matrix proof $M$ containing a clause $D$ with $D \succ C$, and that there is a $\rho$ such that $\rho(C) = \sigma(D)$.
Then $M$ with $D$ exchanged for $C$ also has a spanning set of connections.
\end{lemma}
\begin{proof}
As all variables in $C$ and $D$ are fresh, we can adapt $\sigma$ according to $\rho$. This way, $C$ may be connected to the same literals as $D$. As $\rho(C)$ has the same literals as $\sigma(D)$, we neither add additional paths that must be closed, nor do we prevent other clauses connecting to $C$ because we dropped the respective literal.
\end{proof}
\begin{corollary}[Instance Symmetry Completeness]
Pruning models containing such clauses $D$ remains complete.
\end{corollary}

\subsection{Substitution Symmetry}
Symmetry appears in substitutions applied to different copies of the same clause.
\begin{example}
Consider a literal in two copies of the same clause, $L[x]$ and $L[y]$. 
Assume that all attempts with $\sigma(x) = a$ and $\sigma(y) = b$ fail.
Nothing prevents trying again with all connections ``flipped'' to the other clause and $\sigma(x) = b$ and $\sigma(y) = a$, introducing an exponential number of branches in the worst case.
\end{example}
We enforce an ordering on the substitutions applied to the \emph{variables in copies of the same clause}.
This ordering of terms should be stable under substitution and orient as many terms as possible, but need not have the subterm property and therefore may not be a reduction ordering~\cite{term-rewriting-and-all-that}.
We suggest the following order.
Assume an arbitrary total ordering $\prec$ over function symbols.
Define $f(\bar{t}) \prec g(\bar{s})$ iff (i) $f \prec g$ or (ii) $f = g$ and $\bar{t} \prec \bar{s}$.
Sequences of terms $\bar{t} \prec \bar{s}$ are compared lexicographically.
Let $\bar{x}$ be the variables occurring left-to-right in clause $C$.
Given two instances $C^i$ and $C^j$ of the same clause with $i < j$, we may enforce that $\sigma(\bar{x}_i) \nsucceq \sigma(\bar{x}_j)$ to avoid symmetries over clause substitutions.
\begin{lemma}[Spanning Order]\label{lem:reorder}
Suppose $M$ has a spanning set of connections and contains two copies $C^i$ and $C^j$ of the same clause.
Then there is a spanning set of connections that satisfies $\sigma(\bar{x}_i) \nsucceq \sigma(\bar{x}_j)$.
\end{lemma}
\begin{proof}
If this condition does not already hold, we have $\sigma(\bar{x}_i) \succeq \sigma(\bar{x}_j)$.
Duplicate clauses are already eliminated, so in fact $\sigma(\bar{x}_i) \succ \sigma(\bar{x}_j)$.
Now ``swap'' $C^i$ and $C^j$ by exchanging their connections to obtain a new spanning set of connections and consistent substitution $\sigma'$.
Necessarily, $\sigma'(\bar{x}_i) \prec \sigma'(\bar{x}_j)$.
\end{proof}
By iterated application of Lemma~\ref{lem:reorder}, it is possible to ``reorder'' any spanning set of connections into another that respects the order.
\begin{corollary}[Substitution Symmetry Completeness]
Enforcing an ordering on substitution of variables in copies of the same clause remains complete.
\end{corollary}

\section{Implementation and Experiments}
\label{sec:implementation}\label{sec:implementationOpt}\label{sec:experiments}
\newcolumntype{Y}{>{\centering\arraybackslash}X}

\begin{figure}[t]
    \centering
    \begin{tabularx}{\textwidth}{|l|Y|Y|Y|Y|Y|Y|Y|}
    \hline
    Solver & \upCoP{}$_{{\tt SAT}}$ & \upCoP{}$_{{\tt SMT}}$ & \upCoP{}$_{{\tt SAT}}$ & \upCoP{}$_{{\tt SMT}}$ & \upCoP{}$_{{\tt SAT}}$ & \upCoP{}$_{{\tt SMT}}$ & \meanCoP{} \\
    \hline
    Enc. & \multicolumn{2}{|c|}{$\mathcal{E}_M$} & \multicolumn{2}{|c|}{$\mathcal{E}_U$} & \multicolumn{2}{|c|}{$\mathcal{E}_H$} & \\
    \hline
    Solved & 928 & 855 & 1,152 & 1,055 & 1,272 & 1,264 & 1,972 \\
    \hline
    Unique & 27 & 20 & 109 & 93 & 105 & 76 & 551 \\
    \hline
    \end{tabularx}
    \caption{Experimental summary. \upCoP{}$_{{\tt SAT}}$ and \upCoP{}$_{{\tt SMT}}$ denote the \upCoP{} versions based on \cadical{} (SAT) and \zzz (SMT). ``Unique'' lists the number of problems solvable by \upCoP{}, but not \meanCoP{}, and vice versa.}
    \label{fig:experiments}
\end{figure}

\paragraph{Implementation.}
We implemented the encodings and optimisations of Sections~\ref{sec:tableau}--\ref{sec:optimisations} in our new prototype  \upCoP{}\footnote{at 
\url{https://github.com/CEisenhofer/UPCoP}}.
\upCoP{} uses either the user-propagator of the \cadical{} SAT solver~\cite{ipasir-up} or of the \zzz{} SMT solver~\cite{user-propagation}. 

Due to our encoding via SAT variables,  \upCoP{}  can choose an arbitrary atom to be assigned in case of a decision. However, such an arbitrary atom selection conflicts with the notion of goal-directedness. 
Consider, for example,  the initial constraint~\eqref{eq:start:matrix} that at least one selector of a conjecture clause has to be taken. However, this constraint does not require that the solver start by choosing one of those clauses, but only that one start clause is selected in the final model.
As a result, the solver might pick an arbitrary selector and derive a submatrix from it that is not connectable to any conjecture. A similar situation may occur when propagating all possible candidates in~\eqref{eq:matrixConnect}.
We therefore force \upCoP{} always to pick one of the currently relevant options, rather than leaving the choice to the internal heuristics of \upCoP{}. Further, we make \upCoP{} to assign currently irrelevant selectors false, rather than true.

We note that \upCoP{} might spend significant reasoning time upon term equalities and orderings not required by any selected clause in the matrix. To overcome this issue, \upCoP{} associates every non-ground term uniquely to the clause copy it originates from. This way, we process (in)equalities reported by those non-ground terms only in case the respective selector is assigned true.

\paragraph{Experimental Setup.} 
We evaluated \upCoP{} on the TPTP benchmark set~\cite{TPTP} of first-order problems using an AMD EPYC 7502 clocked at roughly $1.8$GHz, a limit of $16$GB memory per run, and a $30$s timeout. 
We considered all $6468$ provable first-order problems of the TPTP 8.2.0 problem repository by translating them upfront to the SMT-LIB input format~\cite{BarFT-SMTLIB}.

We evaluated \upCoP{} with three encodings, as follows. First, we used the encodings $\mathcal{E}_M$ and $\mathcal{E}_U$ of Sections~\ref{sec:matrix}--\ref{sec:unsatCore}, respectively. Further, we considered a hybrid approach $\mathcal{E}_H$ which combines $\mathcal{E}_M$ and $\mathcal{E}_U$ in a way such that the capacity of every (non-ground) clause is increased on each failure, but the set of clauses in the matrix is limited by the number of copies available, rather than by a strict upper bound.
For each encoding, we separately evaluated \upCoP{} with user-propagators of \cadical{} and \zzz. 
Our experimental results thus report on six different settings of \upCoP{}, shown in Figure~\ref{fig:experiments}.
In addition, we compared \upCoP{} with the default (complete) mode of \meanCoP{}~\cite{curiously-effective}. 

\paragraph{Experimental Analysis.}   
Our results are summarised in Figure~\ref{fig:experiments}, with all versions of \upCoP{} solving altogether 1,601 problems. While \upCoP{} solves less problems than \meanCoP{}, we emphasize that 
\upCoP{} solves $179$ problems from all different kinds of categories that \meanCoP{} cannot. There are many different reasons for this. There are mainly two reasons apart from the ``non-determinism'' in the choice of which literals to connect. The first one is that some matrix proofs found by \upCoP{} are in fact decently smaller than the respective tableau-shaped proof of \meanCoP{}. The second one is that the unsat core-based encodings indeed sometimes deduce that certain clauses are completely irrelevant, and consequently ignore them, resulting in a reasonable speed-up.

Counter-intuitively, \upCoP{} usually does not get stuck in eliminating non-spanning connections. Most cases in which \upCoP{} does not terminate are because \upCoP{} is too slow to find a fully-connected matrix due to the previously discussed points. As soon as a fully-connected matrix is found, it becomes spanning after only a few rejected matrices. Although the spanningness check is co-NP complete, we did not encounter cases where it took a significant time.

We note that $\mathcal{E}_M$ terminates quickly on inputs with small (potentially non-tree-shaped) proofs. $\mathcal{E}_U$ is best suited for large problems with redundant unconnectable clauses, as we will not generate further copies of the redundant clauses. $\mathcal{E}_H$ best pays off on problems whose proofs use most clauses from the input, but each of them only a few times. 

\paragraph{Discussion.} 
When analysing our experimental results, we note the following aspects for further improvements in \upCoP{}:  
\begin{itemize}
    \item The SAT/SMT solvers used by \upCoP{} repeatedly learn additional facts about the underlying proof systems, mostly using new Boolean variables. Hence, the solvers might decide to split on these Boolean variables. This can slow down proof search, as the solver might put a lot of effort into learning facts about partial proofs that cannot be extended to a proper proof. Further, some of these Boolean assignments represent partial proofs that a naive proof enumeration cannot even reach. 
    \item Learning clauses coming from those unfruitful proof spaces also do not necessarily combine: for example, the solver can consider multiple submatrices simultaneously, even though these submatrices can never be combined. The assumption that conflict learning and SAT heuristics will guide the solver towards a more focused search has not been vindicated.
    Postponing the propagation of possible connections from submatrices not attached to the main matrix helps to reduce the number of Boolean variables and the overhead of propagating potentially unused constraints, but requires additional expensive book-keeping.
   \item Adjusting the SAT/SMT solvers' decision heuristic to prefer Boolean variables relevant for the current proof state might be beneficial~\cite{RelevancyProapgation}. Yet, such adjustments harm efficient heuristics for variable selection within the solver.
   \item The encoding $\mathcal{E}_T$ suffers because the conflicts mainly learned depend on the precise paths in the tableau. Similarly, conflicts learned about the Boolean variables in the matrix encodings $\mathcal{E}_M$ and $\mathcal{E}_U$ do not generalise well from one copy to the other in a way such that multiple copies of the same clause are considered in an isomorphic way numerous times. 
   \item Additionally, SAT/SMT solvers incur non-negligible overhead from propagation loops for processing clauses. While simulating a single rule via SAT, a naive proof enumeration solver can process multiple steps within that time.
\end{itemize}

As shown and discussed before, one way of tackling most of those problems is to implement a custom variant of conflict learning that does not have the discussed behaviours that are mostly inherent to modern SAT solvers. Based on these results, we build a further system \textsf{hopCoP}~\cite{learning} that implements such a custom conflict learning engine for the tableaux form of connection calculus.

\section{Related Work}
Most first-order theorem provers employ a variety of ground reasoning techniques, predominantly SAT and SMT solvers.
Here we must mention the family of instance-based methods~\cite{instance-based-methods}: grounding a set of first-order clauses in the hope that they become unsatisfiable, which can be employed with a dedicated calculus~\cite{InstGen} or alongside an existing system~\cite{SATCoP,e-grounding}.
In the other direction, SMT solvers often integrate quantifier instantiation into satisfiability routines~\cite{mbqi,e-matching}.
Ground reasoning can also be used for many other combinatorial tasks in first-order theorem provers~\cite{uses-of-sat-in-vampire}, such as keeping track of clause splitting~\cite{AVATAR}, detecting subsumption~\cite{sat-subsumption}, or determining if inferences are applicable~\cite{sat-subsumption-resolution}.
MACE-style finite model builders~\cite{mace} employ SAT solving to determine whether a set of clauses is satisfiable, assuming a fixed-size finite model, using symmetry-breaking~\cite{symmetry-breaking-fmb}.

Apart from our other approach~\cite{learning} using a custom conflict learning engine and restricting ourselves to the direct encoding of proof objects, the ChewTPTP system in both its SAT~\cite{chewtptp-sat} and SMT~\cite{chewtptp-smt} incarnations is the closest existing approach to theorem proving via satisfiability solving.
ChewTPTP encodes constraints for a closed connection tableau completely ahead of time, then passes the resulting constraints to a SAT or SMT solver. The higher-order systems Satallax~\cite{Satallax}  and its fork Lash~\cite{lash}  also reduce proof search to the propositional level. However, in contrast to \upCoP{}, when the Satallax's SAT solver reports unsatisfiability, a proof has been found.
We previously published an early version of our ideas in a more general setting~\cite{upCoP}.

\section{Conclusion}
We encode first-order connection calculi as a propositional problem, improve our SAT encodings for matrix proofs, and guide iterative deepening using unsat cores. We introduce several optimisations to prune symmetries and eliminate unnecessary branches during proof search. Initial experiments with our new \upCoP{} solver show that our encoding can be used to solve certain problems, other connection solvers fail to solve. Further improvements may come with applying custom garbage-collection and variable selection heuristic strategies in \upCoP{} due to the discussed problems encountered during solving.

\subsubsection{Acknowledgements.} This research was funded in whole or in part by the  ERC Consolidator Grant ARTIST 101002685, the ERC Proof of Concept Grant LEARN 101213411, the TU Wien Doctoral College SecInt, the FWF SpyCoDe Grant 10.55776/F85,  the WWTF grant ForSmart   10.47379/ICT22007, and the Amazon Research Award 2023 QuAT.

\subsubsection{Disclosure of Interests.}
The authors have no competing interests to declare that are relevant to the content of this article.

\bibliographystyle{splncs04}
\bibliography{main}

@article{DBLP:journals/jar/Andrews89,
  author       = {Peter B. Andrews},
  title        = {{On Connections and Higher-Order Logic}},
  journal      = {J. Autom. Reason.},
  volume       = {5},
  number       = {3},
  pages        = {257--291},
  year         = {1989},
  url          = {https://doi.org/10.1007/BF00248320},
  doi          = {10.1007/BF00248320},
  }

@article{Galmiche2000, 
title = {{Connection Methods in Linear Logic and Proof Nets Construction}},
journal = {Theoretical Computer Science},
volume = {232},
number = {1},
pages = {231-272},
year = {2000},
issn = {0304-3975},
doi = {https://doi.org/10.1016/S0304-3975(99)00176-0},
url = {https://www.sciencedirect.com/science/article/pii/S0304397599001760},
author = {D. Galmiche},
}

@inproceedings{DBLP:conf/stacs/Goubault94,
  author       = {Jean Goubault},
  title        = {The Complexity of Resource-Bounded First-Order Classical Logic},
  booktitle    = {{STACS}},
  series       = {{LNCS}},
  pages        = {59--70},
  year         = {1994},
  doi          = {10.1007/3-540-57785-8\_131}
}

@inproceedings{leanCoP,
  author       = {Jens Otten},
  title        = {\textsf{leanCoP} 2.0 and \textsf{ileanCoP} 1.2: High Performance Lean Theorem Proving in Classical and Intuitionistic Logic (System Descriptions)},
  booktitle    = {{IJCAR}},
  series       = {{LNCS}},
  pages        = {283--291},
  year         = {2008},
  doi          = {10.1007/978-3-540-71070-7\_23}
}

@inproceedings{curiously-effective,
  author       = {Michael F{\"{a}}rber},
  title        = {A Curiously Effective Backtracking Strategy for Connection Tableaux},
  pages        = {23--40},
  booktitle     = {{AReCCa}}, 
  year = {2023}, 
  url={https://ceur-ws.org/Vol-3613/},
}

@inproceedings{upCoP,
  author       = {Eisenhofer, Clemens and Kov{\'a}cs, Laura and Rawson, Michael},
  title        = {Embedding the Connection Calculus in Satisfiability Modulo Theories},
  pages        = {54--63},
  booktitle     = {{AReCCa}}, 
  year = {2023}, 
  url={https://ceur-ws.org/Vol-3613/},
}

@inproceedings{Vampire,
  author       = {Laura Kov{\'{a}}cs and Andrei Voronkov},
  title        = {First-Order Theorem Proving and \textsc{Vampire}},
  booktitle    = {{CAV}},
  series       = {{LNCS}},
  pages        = {1--35},
  year         = {2013},
  doi          = {10.1007/978-3-642-39799-8\_1}
}

@inproceedings{eprover,
  author       = {Stephan Schulz and
                  Simon Cruanes and
                  Petar Vukmirovic},
  title        = {Faster, Higher, Stronger: {E} 2.3},
  booktitle    = {{CADE}},
  series       = {{LNCS}},
  pages        = {495--507},
  year         = {2019},
  doi          = {10.1007/978-3-030-29436-6\_29}
}

@inproceedings{SATCoP,
  author       = {Michael Rawson and
                  Giles Reger},
  title        = {Eliminating Models During Model Elimination},
  booktitle    = {{TABLEAUX}},
  series       = {{LNCS}},
  volume       = {12842},
  pages        = {250--265},
  year         = {2021},
  doi          = {10.1007/978-3-030-86059-2\_15}
}

@inproceedings{AVATAR,
  author       = {Andrei Voronkov},
  title        = {{AVATAR:} The Architecture for First-Order Theorem Provers},
  booktitle    = {{CAV}},
  series       = {{LNCS}},
  volume       = {8559},
  pages        = {696--710},
  year         = {2014},
  doi          = {10.1007/978-3-319-08867-9\_46}
}

@book{tableauxHandbook,
  title={Handbook of tableau methods},
  author={D'Agostino, Marcello and Gabbay, Dov M and H{\"a}hnle, Reiner and Posegga, Joachim},
  year={2013},
  publisher={Springer Science \& Business Media}
}

@book{Bibel,
  series = {Artificial intelligence},
  isbn = {3528185201},
  year = {1987},
  title = {Automated theorem proving},
  edition = {2., rev. ed.},
  publisher = {Vieweg},
  author = {Bibel, Wolfgang}
}

@incollection{taxonomy,
  author       = {Maria Paola Bonacina},
  title        = {A Taxonomy of Theorem-Proving Strategies},
  booktitle    = {Artificial Intelligence Today: Recent Trends and Developments},
  series       = {{LNCS}},
  volume       = {1600},
  pages        = {43--84},
  publisher    = {Springer},
  year         = {1999},
  doi          = {10.1007/3-540-48317-9\_3}
}

@inproceedings{comparison-of-proof-methods,
  title={Comparison of proof methods},
  author={Bibel, Wolfgang},
  pages={119-132},
 booktitle     = {{AReCCa}}, 
year = {2023}, 
url={https://ceur-ws.org/Vol-3613/},
}

@book{handbook-of-satisfiability,
  title        = {Handbook of Satisfiability - Second Edition},
  editor       = {Armin Biere and
                  Marijn Heule and
                  Hans van Maaren and
                  Toby Walsh},
  series       = {Frontiers in Artificial Intelligence and Applications},
  volume       = {336},
  year         = {2021},
  publisher    = {{IOS} Press},
  doi          = {10.3233/FAIA336},
  isbn         = {978-1-64368-160-3}
}

@book{smullyan,
	author = {Raymond Merrill Smullyan},
	title = {First-Order Logic},
	year = {1968},
    publisher={Springer}
}

@article{setheo,
  author       = {Reinhold Letz and
                  Johann Schumann and
                  Stefan Bayerl and
                  Wolfgang Bibel},
  title        = {{SETHEO:} {A} High-Performance Theorem Prover},
  journal      = {J. Autom. Reason.},
  volume       = {8},
  number       = {2},
  pages        = {183--212},
  year         = {1992},
  doi          = {10.1007/BF00244282}
}

@inproceedings{spass,
  author       = {Christoph Weidenbach and
                  Dilyana Dimova and
                  Arnaud Fietzke and
                  Rohit Kumar and
                  Martin Suda and
                  Patrick Wischnewski},
  title        = {{SPASS} Version 3.5},
  booktitle    = {{CADE}},
  series       = {{LNCS}},
  volume       = {5663},
  pages        = {140--145},
  year         = {2009},
  doi          = {10.1007/978-3-642-02959-2\_10}
}

@inproceedings{ipasir-up,
  author       = {Katalin Fazekas and
                  Aina Niemetz and
                  Mathias Preiner and
                  Markus Kirchweger and
                  Stefan Szeider and
                  Armin Biere},
  title        = {{IPASIR-UP:} User Propagators for {CDCL}},
  booktitle    = {{SAT}},
  series       = {LIPIcs},
  volume       = {271},
  pages        = {8:1--8:13},
  year         = {2023},
  doi          = {10.4230/LIPICS.SAT.2023.8}
}

@inproceedings{user-propagation,
  author       = {Nikolaj S. Bj{\o}rner and
                  Clemens Eisenhofer and
                  Laura Kov{\'{a}}cs},
  title        = {Satisfiability Modulo Custom Theories in {Z3}},
  booktitle    = {{VMCAI}},
  series       = {{LNCS}},
  volume       = {13881},
  pages        = {91--105},
  year         = {2023},
  doi          = {10.1007/978-3-031-24950-1\_5}
}

@article{instance-based-methods,
  author       = {Peter Baumgartner and
                  Evgenij Thorstensen},
  title        = {Instance Based Methods -- {A} Brief Overview},
  journal      = {K{\"{u}}nstliche Intell.},
  volume       = {24},
  number       = {1},
  pages        = {35--42},
  year         = {2010},
  doi          = {10.1007/S13218-010-0002-X}
}

@inproceedings{sat-subsumption,
  author       = {Jakob Rath and
                  Armin Biere and
                  Laura Kov{\'{a}}cs},
  title        = {First-Order Subsumption via {SAT} Solving},
  booktitle    = {{FMCAD}},
  pages        = {160--169},
  year         = {2022},
  doi          = {10.34727/2022/ISBN.978-3-85448-053-2\_22}
}

@inproceedings{DBLP:conf/jelia/Tinelli02,
  author       = {Cesare Tinelli},
  title        = {{A DPLL-Based Calculus for Ground Satisfiability Modulo Theories}},
  booktitle    = {{JELIA}},
  pages        = {308--319},
  year         = {2002},
  doi          = {10.1007/3-540-45757-7\_26}
}

@inproceedings{DBLP:conf/sat/BiereFW23,
  author       = {Armin Biere and
                  Nils Froleyks and
                  Wenxi Wang},
  title        = {{CadiBack: Extracting Backbones with CaDiCaL}},
  booktitle    = {SAT},
  series       = {LIPIcs},
  volume       = {271},
  pages        = {3:1--3:12},
  year         = {2023},
  doi          = {10.4230/LIPICS.SAT.2023.3},
}

@article{e-grounding,
  title={Light-weight integration of {SAT} solving into first-order reasoners -- first experiments},
  author={Schulz, Stephan},
  journal={{Vampire}},
  pages={9--19},
  year={2017}
}

@inproceedings{instgen,
  author       = {Konstantin Korovin},
  title        = {{Inst-Gen} - {A} Modular Approach to Instantiation-Based Automated Reasoning},
  booktitle    = {Programming Logics - Essays in Memory of Harald Ganzinger},
  series       = {{LNCS}},
  volume       = {7797},
  pages        = {239--270},
  year         = {2013},
  doi          = {10.1007/978-3-642-37651-1\_10}
}

@inproceedings{mace,
  title={New techniques that improve {MACE}-style finite model finding},
  author={Claessen, Koen and S{\"o}rensson, Niklas},
  booktitle={Proceedings of the CADE-19 Workshop: Model Computation-Principles, Algorithms, Applications},
  pages={11--27},
  year={2003}
}

@inproceedings{symmetry-breaking-fmb,
  author       = {Giles Reger and
                  Martin Riener and
                  Martin Suda},
  title        = {Symmetry Avoidance in {MACE}-Style Finite Model Finding},
  booktitle    = {{FroCoS}},
  series       = {{LNCS}},
  volume       = {11715},
  pages        = {3--21},
  year         = {2019},
  doi          = {10.1007/978-3-030-29007-8\_1}
}

@inproceedings{uses-of-sat-in-vampire,
  author       = {Giles Reger and
                  Martin Suda},
  title        = {The Uses of {SAT} Solvers in \textsc{Vampire}},
  booktitle    = {{Vampire}},
  series       = {EPiC Series in Computing},
  volume       = {38},
  pages        = {63--69},
  year         = {2015},
  doi          = {10.29007/4W68}
}

@inproceedings{chewtptp-sat,
  author       = {Todd Deshane and
                  Wenjin Hu and
                  Patty Jablonski and
                  Hai Lin and
                  Christopher Lynch and
                  Ralph Eric McGregor},
  title        = {Encoding First Order Proofs in {SAT}},
  booktitle    = {{CADE}},
  series       = {{LNCS}},
  volume       = {4603},
  pages        = {476--491},
  year         = {2007},
  doi          = {10.1007/978-3-540-73595-3\_35}
}

@inproceedings{chewtptp-smt,
  author       = {Jeremy Bongio and
                  Cyrus Katrak and
                  Hai Lin and
                  Christopher Lynch and
                  Ralph Eric McGregor},
  title        = {Encoding First Order Proofs in {SMT}},
  booktitle    = {{SMT}},
  series       = {{ENTCS}},
  volume       = {198},
  pages        = {71--84},
  year         = {2007},
  doi          = {10.1016/J.ENTCS.2008.04.081}
}

@incollection{handbook-ar-nf,
  author       = {Matthias Baaz and
                  Uwe Egly and
                  Alexander Leitsch},
  title        = {Normal Form Transformations},
  booktitle    = {Handbook of Automated Reasoning (in 2 volumes)},
  publisher    = {Elsevier and {MIT} Press},
  pages        = {273--333},
  year         = {2001},
  doi          = {10.1016/B978-044450813-3/50007-2}
}

@incollection{handbook-ar-paramodulation,
  author       = {Robert Nieuwenhuis and
                  Albert Rubio},
  title        = {Paramodulation-Based Theorem Proving},
  booktitle    = {Handbook of Automated Reasoning (in 2 volumes)},
  pages        = {371--443},
  publisher    = {Elsevier and {MIT} Press},
  year         = {2001},
  doi          = {10.1016/B978-044450813-3/50009-6}
}

@incollection{handbook-ar-model-elimination,
  author       = {Reinhold Letz and
                  Gernot Stenz},
  title        = {Model Elimination and Connection Tableau Procedures},
  publisher    = {Elsevier and {MIT} Press},
  booktitle    = {Handbook of Automated Reasoning (in 2 volumes)},
  pages        = {2015--2114},
  year         = {2001},
  doi          = {10.1016/B978-044450813-3/50030-8}
}

@article{brand,
  author       = {Daniel Brand},
  title        = {Proving Theorems with the Modification Method},
  journal      = {{SIAM} J. Comput.},
  volume       = {4},
  number       = {4},
  pages        = {412--430},
  year         = {1975},
  doi          = {10.1137/0204036}
}

@inproceedings{using-matings-for-pruning,
  author       = {Reinhold Letz},
  title        = {Using Matings for Pruning Connection Tableaux},
  booktitle    = {{CADE}},
  series       = {{LNCS}},
  volume       = {1421},
  pages        = {381--396},
  year         = {1998},
  doi          = {10.1007/BFB0054273}
}

@inproceedings{matrix-based-constructive,
  author       = {Christoph Kreitz and
                  Jens Otten and
                  Stephan Schmitt and
                  Brigitte Pientka},
  title        = {Matrix-based Constructive Theorem Proving},
  booktitle    = {Intellectics and Computational Logic (to {Wolfgang Bibel} on the occasion
                  of his 60th birthday)},
  series       = {Applied Logic Series},
  volume       = {19},
  pages        = {189--205},
  year         = {2000}
}

@article{matings-in-matrices,
  author       = {Wolfgang Bibel},
  title        = {Matings in Matrices},
  journal      = {Commun. {ACM}},
  volume       = {26},
  number       = {11},
  pages        = {844--852},
  year         = {1983},
  doi          = {10.1145/182.183}
}

@incollection{term-indexing,
  author       = {I. V. Ramakrishnan and
                  R. Sekar and
                  Andrei Voronkov},
  title        = {{Term Indexing}},
  booktitle    = {Handbook of Automated Reasoning (in 2 volumes)},
  pages        = {1853--1964},
  publisher    = {Elsevier and {MIT} Press},
  year         = {2001},
  doi          = {10.1016/b978-044450813-3/50028-x}
}

@article{datatypes,
  author       = {Clark W. Barrett and
                  Igor Shikanian and
                  Cesare Tinelli},
  title        = {An Abstract Decision Procedure for a Theory of Inductive Data Types},
  journal      = {J. Satisf. Boolean Model. Comput.},
  volume       = {3},
  number       = {1-2},
  pages        = {21--46},
  year         = {2007},
  doi          = {10.3233/SAT190028}
}

@inproceedings{e-matching,
  author       = {Leonardo Mendon{\c{c}}a de Moura and Nikolaj S. Bj{\o}rner},
  title        = {Efficient E-Matching for {SMT} Solvers},
  booktitle    = {{CADE}},
  series       = {{LNCS}},
  volume       = {4603},
  pages        = {183--198},
  year         = {2007},
  doi          = {10.1007/978-3-540-73595-3\_13}
}

@inproceedings{mbqi,
  author       = {Yeting Ge and Leonardo Mendon{\c{c}}a de Moura},
  title        = {Complete Instantiation for Quantified Formulas in Satisfiabiliby Modulo Theories},
  booktitle    = {{CAV}},
  series       = {{LNCS}},
  volume       = {5643},
  pages        = {306--320},
  year         = {2009},
  doi          = {10.1007/978-3-642-02658-4\_25}
}

@inproceedings{DBLP:conf/cp/Sinz05,
  author       = {Carsten Sinz},
  title        = {Towards an Optimal {CNF} Encoding of {Boolean} Cardinality Constraints},
  booktitle    = {{CP}},
  series       = {{LNCS}},
  volume       = {3709},
  pages        = {827--831},
  year         = {2005},
  doi          = {10.1007/11564751\_73}
}

@inproceedings{DBLP:conf/cp/BailleuxB03,
  author       = {Olivier Bailleux and Yacine Boufkhad},
  title        = {Efficient {CNF} Encoding of {Boolean} Cardinality Constraints},
  booktitle    = {{CP}},
  series       = {{LNCS}},
  volume       = {2833},
  pages        = {108--122},
  year         = {2003},
  doi          = {10.1007/978-3-540-45193-8\_8}
}

@article{DBLP:journals/tcad/ChaiK05,
  author       = {Donald Chai and Andreas Kuehlmann},
  title        = {A fast pseudo-{Boolean} constraint solver},
  journal      = {{IEEE} Trans. Comput. Aided Des. Integr. Circuits Syst.},
  volume       = {24},
  number       = {3},
  pages        = {305--317},
  year         = {2005},
  doi          = {10.1109/TCAD.2004.842808}
}

@inproceedings{sat-subsumption-resolution,
  author       = {Robin Coutelier and Laura Kov{\'{a}}cs and Michael Rawson and Jakob Rath},
  title        = {{SAT}-Based Subsumption Resolution},
  booktitle    = {{CADE}},
  series       = {{LNCS}},
  volume       = {14132},
  pages        = {190--206},
  year         = {2023},
  doi          = {10.1007/978-3-031-38499-8\_11}
}

@book{term-rewriting-and-all-that,
  author       = {Franz Baader and Tobias Nipkow},
  title        = {Term rewriting and all that},
  year         = {1998},
  publisher    = {Cambridge university press},
  isbn         = {978-0-521-45520-6}
}

@article{DBLP:journals/tc/AloulSM06,
  author       = {Fadi A. Aloul and
                  Karem A. Sakallah and
                  Igor L. Markov},
  title        = {Efficient Symmetry Breaking for {Boolean} Satisfiability},
  journal      = {{IEEE} Trans. Computers},
  volume       = {55},
  number       = {5},
  pages        = {549--558},
  year         = {2006},
  doi          = {10.1109/TC.2006.75}
}

@inproceedings{minimal-unsat-cores,
  author       = {In{\^{e}}s Lynce and
                  Jo{\~{a}}o Marques{-}Silva},
  title        = {On Computing Minimum Unsatisfiable Cores},
  booktitle    = {{SAT}},
  year         = {2004},
  url          = {http://www.satisfiability.org/SAT04/programme/110.pdf}
}

@article{minimal-unsat-cores-smt,
  author       = {Alessandro Cimatti and
                  Alberto Griggio and
                  Roberto Sebastiani},
  title        = {Computing Small Unsatisfiable Cores in Satisfiability Modulo Theories},
  journal      = {J. Artif. Intell. Res.},
  volume       = {40},
  pages        = {701--728},
  year         = {2011},
  doi          = {10.1613/JAIR.3196}
}

@article{cegar,
  author       = {Edmund M. Clarke and
                  Orna Grumberg and
                  Somesh Jha and
                  Yuan Lu and
                  Helmut Veith},
  title        = {Counterexample-guided abstraction refinement for symbolic model checking},
  journal      = {J. {ACM}},
  volume       = {50},
  number       = {5},
  pages        = {752--794},
  year         = {2003},
  doi          = {10.1145/876638.876643}
}

@Article{TPTP,
    author       = {Sutcliffe, G.},
    year         = {2017},
    title        = {The {TPTP} Problem Library and Associated Infrastructure. From {CNF} to {TH0}, {TPTP v6.4.0}},
    journal      = {J. Autom. Reason.},
    volume       = {59},
    number       = {4},
    pages        = {483-502}
}

@inproceedings{DBLP:conf/lpar/Otten12,
  author       = {Jens Otten},
  title        = {Implementing Connection Calculi for First-order Modal Logics},
  booktitle    = {{IWIL}},
  series       = {EPiC Series in Computing},
  volume       = {22},
  pages        = {18--32},
  publisher    = {EasyChair},
  year         = {2012},
  doi          = {10.29007/82M9}
}

@MISC{BarFT-SMTLIB,
  author =	 {Clark Barrett and Pascal Fontaine and Cesare Tinelli},
  title =	 {The Satisfiability Modulo Theories Library ({SMT-LIB})},
  howpublished = {{\tt www.SMT-LIB.org}},
  year =	 2016,
}

@inproceedings{DBLP:conf/tacas/MouraB08,
  author       = {Leonardo Mendon{\c{c}}a de Moura and
                  Nikolaj S. Bj{\o}rner},
  title        = {{Z3:} An Efficient {SMT} Solver},
  booktitle    = {{TACAS}},
  series       = {{LNCS}},
  volume       = {4963},
  pages        = {337--340},
  publisher    = {Springer},
  year         = {2008},
  doi          = {10.1007/978-3-540-78800-3\_24}
}

@inproceedings{CaDiCal,
    author    = {Armin Biere and Katalin Fazekas and Mathias Fleury and Maximillian Heisinger},
    title     = {{CaDiCaL}, {Kissat}, {Paracooba}, {Plingeling} and {Treengeling}
		 Entering the {SAT Competition 2020}},
    pages     = {51--53},
    booktitle = {Proc.~of {SAT Competition} 2020 -- Solver and Benchmark Descriptions},
    series    = {Department of Computer Science Report Series B},
    publisher = {University of Helsinki},
    year      = {2020}
}

@article{RelevancyProapgation,
  title={Relevancy propagation},
  author={de Moura, Leonardo and Bj{\o}rner, Nikolaj},
  journal={Technical Report MSR-TR-2007-140, Microsoft Research, Tech. Rep.},
  year={2007},
  publisher={Citeseer}
}

@inproceedings{learning,
    author = {Michael Rawson and Clemens Eisenhofer and Laura Kov{\'{a}}cs},
    title = {Constraint Learning for Non-Confluent Proof Search},
    booktitle    = {{TABLEAUX}},
    volume       = {15980},
    series       = {{LNCS}},
    year         = {2025},
    pages        = {103--119},
    doi          = {10.1007/978-3-032-06085-3\_6}
}

@article{lash,
  author       = {Chad E. Brown and
                  Cezary Kaliszyk},
  title        = {{L}ash 1.0 (System Description)},
  journal      = {CoRR},
  volume       = {abs/2205.06640},
  year         = {2022},
  doi          = {10.48550/ARXIV.2205.06640},
  eprinttype    = {arXiv},
  eprint       = {2205.06640}
}

@inproceedings{Satallax,
  author       = {Michael F{\"{a}}rber and
                  Chad E. Brown},
  title        = {Internal Guidance for Satallax},
  booktitle    = {{IJCAR}},
  series       = {{LNCS}},
  volume       = {9706},
  pages        = {349--361},
  publisher    = {Springer},
  year         = {2016},
  doi          = {10.1007/978-3-319-40229-1\_24}
}
\end{document}